\documentclass[12pt]{amsart}
\usepackage{amscd,amssymb,bm}
\usepackage{hyperref}
\usepackage{verbatim}

\usepackage{enumerate}

\usepackage{url}

\numberwithin{equation}{section}

\newcommand{\ndash}{\nobreakdash-\hspace{0pt}}
\newcommand{\Ndash}{\nobreakdash--}

\newcommand{\ii}{{\mathrm{i}}}
\newcommand{\dd}{{\mathrm{d}}}

\newcommand{\T}{{\mathbb T}}
\newcommand{\superx}{{\mathsf x}}
\newcommand{\supereta}{{\mathsf e}}
\newcommand{\z}{\bar z}
\newcommand{\J}{{\bar \jmath}}
\newcommand{\I}{{\bar \imath}}

\newcommand{\Seff}{S_{\mathrm{eff}}}
\newcommand{\tSeff}{{\Tilde S_{\mathrm{eff}}}}

\newcommand{\mr}{\mathrm}
\newcommand{\HH}{\mathcal{H}}
\newcommand{\PP}{\mathcal{P}}

\DeclareMathOperator{\Maps}{Map}

\newcommand{\LL}{\mathcal{L}}

\newcommand{\FF}{\mathcal{F}}

\newcommand{\bt}{\bullet}
\newcommand{\ra}{\rightarrow}
\newcommand{\be}{\begin{equation}}
\newcommand{\ee}{\end{equation}}

\newcommand{\calL}{{\mathcal{L}}}
\newcommand{\calW}{{\mathcal{W}}}
\newcommand{\calZ}{{\mathcal{Z}}}
\newcommand{\calF}{{\mathcal{F}}}

\newcommand{\EE}{\mathrm{e}}

\DeclareMathOperator{\Der}{Der}
\DeclareMathOperator{\Map}{Map}
\DeclareMathOperator{\Div}{div}

\DeclareMathOperator{\Pf}{Pf}

\newcommand{\id}{\mathrm{id}}
\newcommand{\ev}{\mathrm{ev}}

\DeclareMathOperator{\str}{STr}

\newcommand{\frX}{\mathfrak{X}}

\DeclareMathOperator{\End}{End}

\newtheorem{Thm}{Theorem}[section]

\newtheorem{Lem}[Thm]{Lemma}

\newtheorem*{Thm*}{Theorem}
\newtheorem*{Lem*}{Lemma}

\theoremstyle{remark}
\newtheorem{Rem}[Thm]{Remark}
\newtheorem*{Ack}{Acknowledgment}

\newtheorem*{Rem*}{Remark}

\theoremstyle{definition}

\newcommand{\bbR}{{\mathbb{R}}}

\newcommand{\de}{\partial}

\newcommand{\calH}{\mathcal{H}}

\newcommand{\calO}{\mathcal{O}}
\newcommand{\calM}{\mathcal{M}}
\newcommand{\calV}{\mathcal{V}}

\newcommand{\sfeta}{\boldsymbol{\eta}}

\newcommand{\sfA}{{\mathsf{A}}}
\newcommand{\sfB}{{\mathsf{B}}}

\newcommand{\sfX}{{\mathsf{X}}}

\def\gpd{\,\lower1pt\hbox{$\longrightarrow$}\hskip-.24in\raise2pt
               \hbox{$\longrightarrow$}\,}

\let\Tilde=\widetilde

\let\Hat=\widehat

\newcommand\qq{}
\newcommand\cmp[1]{{\qq Commun.\ Math.\ Phys.\ \bf #1}}

\newcommand\mpl[1]{{\qq Mod.\ Phys.\ Lett.\ \bf #1}}

\newcommand\lmp[1]{{\qq Lett.\ Math.\ Phys.\ \bf #1}}

\newcommand\ijmp[1]{{\qq Int.\ J. Mod.\ Phys.\ \bf #1}}

\newcommand\anp[1]{{\qq Ann.\ Phys.\ \bf #1}}

\newcommand\jdg[1]{{\qq J.\ Diff.\ Geom.\ \bf #1}}

\newcommand\ptps[1]{{\qq Prog.\ Theor.\ Phys.\ Suppl.\ \bf #1}}

\newcommand\proma[1]{{\qq Progress in Mathematics \bf #1}}
\newcommand\conm[1]{{\qq Cont.\  Math.\  \bf #1}}
\newcommand\jgp[1]{{\qq J. Geom.\ Phys.\ \bf #1}}

\begin{document}
\title{The Poisson sigma model on closed surfaces}


\author[F.~Bonechi]{Francesco Bonechi}
\address{I.N.F.N. and Dipartimento di Fisica, Via G. Sansone 1, I-50019 Sesto Fiorentino - Firenze, Italy}
\email{bonechi@fi.infn.it}

\author[A.~S.~Cattaneo]{Alberto~S.~Cattaneo}
\address{Institut f\"ur Mathematik, Universit\"at Z\"urich,
Winterthurerstrasse 190, CH-8057 Z\"urich, Switzerland}
\email{alberto.cattaneo@math.uzh.ch}

\author[P. Mnev]{Pavel Mnev}
\address{Institut f\"ur Mathematik, Universit\"at Z\"urich,
Winterthurerstrasse 190, CH-8057 Z\"urich, Switzerland}
\email{pmnev@pdmi.ras.ru}

\thanks{A.~S.~C. acknowledges partial support of SNF Grant No.~200020-131813/1.
P.~M. acknowledges partial support of RFBR Grants
Nos.~11-01-00570-a and
11-01-12037-ofi-m-2011
and of SNF Grant No.~200021-137595.}


\begin{abstract}
Using methods of formal geometry, 
the Poisson sigma model on a closed surface is studied in perturbation theory.
The effective action,
as a function on vacua,
is shown to have no quantum corrections if
the surface is a torus or if the Poisson structure is regular and unimodular (e.g., symplectic).
In the case of a K\"ahler structure or of a trivial Poisson structure,
the partition function on the torus is shown to be the Euler characteristic of the target; some evidence is given
for  this to happen more generally.
The methods of formal geometry introduced in this paper might be applicable to other sigma models, at least of the AKSZ type.
\end{abstract}

\maketitle

\tableofcontents

\section{Introduction}\label{intro}
In this note we study 
 the Poisson sigma model \cite{I,SS} with worldsheet a connected, closed surface $\Sigma$.
To do so we treat the Poisson structure on the target manifold $M$
as a perturbation and expand around the vacua (a.k.a.\ zero modes)
of the unperturbed action.\footnote{What we compute is then $\langle\EE^{\frac\ii\hbar S_\pi}\rangle_0$
where $S_\pi$ is the interaction part depending on the target Poisson structure $\pi$ and
$\langle\ \rangle_0$ denotes the expectation value for the Poisson sigma model with zero Poisson structure.}
As a critical point of the latter in particular contains a constant map,
we have first to localize around its image $x\in M$.
To glue perturbations around different points $x$, we use formal geometry \cite{GK}. Our first result is that the perturbative effective action (as a function on the moduli space of vacua for the unperturbed theory) has no quantum corrections if
$\Sigma$ is the torus or if the Poisson structure is regular and unimodular (e.g., symplectic). In the former case, under the further assumption that the Poisson structure is K\"ahler, we can also perform the integration over vacua and show that the partition function is the Euler characteristic of $M$.
For a general Poisson structure we can use worldsheet supersymmetry to regularize the effective action\footnote{We prove that the regularized effective action
does not depend on the regularization as long as one is present. However, in principle this is not the same theory as the non regularized one.}
and study it like in \cite{Wi}; this argument is however a bit formal unless some extra
conditions on the Poisson structure are assumed.

Notice that on the torus we need not assume unimodularity. For other genera,
the requirement of unimodularity was first remarked in \cite{BZ} where the leading term of the effective action on the sphere was also computed.

The techniques presented in this note, in particular the way of using formal geometry to get a global effective action, should be applicable to other field theories, in particular of the AKSZ type \cite{AKSZ}. The techniques of subsection~\ref{s:axial} and of subsections~\ref{ss-rea} and \ref{ss-rsf} should also extend to higher dimensional AKSZ theories in which the source
manifold is a Cartesian product with a torus.

The torus case may also be understood as follows.
Recall that the BV (Batalin--Vilkovisky) action for the Poisson sigma model can be given in terms of the AKSZ construction \cite{CF-AKSZ}. It is a function on the infinite dimensional graded manifold
$\Map(T[1]\Sigma,T^*[1]M)$. On a cylinder $\Sigma=S^1\times I$, the partition function should be interpreted as an operator on the Hilbert space associated to the boundary $S^1$.
As the theory is topological, this operator is the identity and the partition function on the torus is just
its supertrace.
Now, 
in the case of trivial Poisson structure,
the BFV (Batalin--Fradkin--Vilkovisky) reduced phase space associated to the boundary is the graded symplectic manifold $T^*T^*[1]M=T^*T[-1]M$. If we choose  the vertical polarization in the second presentation, the Hilbert space will be $C^\infty(T[-1]M)$, i.e., the
de~Rham complex with opposite grading. It is then to be expected that the partition function on the torus should be the Euler characteristic of $M$. In the perturbative computation, however, the final result is usually of the form $0\cdot\infty$, but in the K\"ahler case we get an unambiguous answer. We might then think of a K\"ahler structure on $M$, if it exists, as a regularization of the Poisson sigma model with trivial Poisson structure.\footnote{That is, we regularize
$\str_{C^\infty(T[-1]M)}\id=\langle 1\rangle_0$ as
\[
\langle 1\rangle_0 := \lim_{\epsilon\to0}\langle \EE^{\frac{\ii\epsilon}\hbar S_\pi}\rangle_0.
\]
We then show that, in the K\"ahler case, $\langle \EE^{\frac{\ii\epsilon}\hbar S_\pi}\rangle_0$
is independent of $\epsilon$ and equal to the Euler characteristic of $M$.}
Notice that, if such structures exist, they are open dense in the space of all Poisson structures on $M$. Another regularization, which produces the same result,
consists in adding the Hamiltonian functions of the supersymmetry generators for the effective action. Formally, this can even been done before integrating over
fluctuations around vacua.

Finally, notice that apart from the cases mentioned above we do expect the effective action to have quantum corrections. Moreover, the naively computed effective action in formal coordinates might happen not to be global. We show however that
it is always possible to find a quantum canonical transformation which makes it into the Taylor expansion of a global effective action. Its class modulo quantum canonical
transformations is then the well-defined object associated to the theory.

Section~\ref{s-formal} is a crash course in formal geometry (essentially following \cite[\S 2]{CF}). In Section~\ref{s-PSMformal}, we develop the construction of \cite[\S 6]{CF} to define
the effective action in formal coordinates. Next using the results of Section~\ref{s:caxial}, we show that, in the two special cases mentioned above, the effective action has no quantum correction and is the expression in formal coordinates of a global effective action. In Section~\ref{s:pf}, we study the effective action for the case of the torus and perform the computation of the partition function. 
Finally, in Section~\ref{s:global} we study the globalization of the effective action in general.

\begin{Ack}
We thank G.~Felder, T. Johnson-Freyd,
T.~Willwacher and M.~Zabzine
for useful discussions. A.S.C.\ thanks the University of Florence for hospitality.
\end{Ack}

\section{Formal local coordinates}\label{s-formal}
We shortly review the notion of formal local coordinates following the simple introduction of \cite[\S 2]{CF} (for more on formal geometry see \cite{BR,B,GK}).

A generalized exponential map for a manifold $M$ is just a smooth map $\phi\colon U\to M$, where $U$ is some open neighborhood of the zero section of $M$ in $TM$,
$(x,y\in U_x)\mapsto \phi_x(y)$, satisfying $\phi_x(0)=x$ and $\dd_y\phi_x(0)=\id$ $\forall x\in M$.
As an example, one may take the exponential map of a connection.

If $f$ is a smooth function on $M$, then the function $\phi^*f\in C^\infty(U)$ satisfies $\dd (\phi^*f)=\dd f\circ\dd\phi$. Denoting by $\dd_x$ ($\dd_y$) the horizontal (vertical) part of the differential, we then get $\dd_x (\phi^*f)=\dd f\circ\dd_x\phi$ and $\dd_y (\phi^*f)=\dd f\circ\dd_y\phi$. Because of the assumptions on $\phi$, there is an open neighborhood $U'\subset U$
of the zero section of $M$ in $TM$ on which $\dd_y\phi$ is invertible. As a consequence, on $U'$ we have the formula
\begin{equation}\label{e:dphif}
\dd_x(\phi^*f) = \dd_y(\phi^*f)\circ(\dd_y\phi)^{-1}\circ\dd_x\phi.
\end{equation}

Notice that, for each $x$, $\phi_x^*f$ is a smooth function on $U_x$.
By $T\phi_x^*f\in \Hat S T_x^*M$  we then denote its Taylor expansion in the $y\in U_x$\ndash variables around $y=0$.\footnote{Here $\Hat S$ denotes the formal completion of the symmetric algebra.}
In doing this, we associate to $f\in C^\infty(M)$ a section $T\phi^* f$ of $\Hat S T^*M$ over $M$. We may now reinterpret \eqref{e:dphif} as a condition on the section $T\phi^*f$
simply taking Taylor expansions w.r.t.\ $y$ on both sides. Notice that in the definition of $T\phi^*f$ and in the resulting condition only the Taylor coefficients of $\phi$ appears.
We are thus let to considering two generalized exponential maps as equivalent if all their partial derivatives in the vertical directions, for each point of the base $M$, coincide at the zero section.
We call formal exponential map an equivalence class of generalized exponential maps.

If $\phi$ is a formal exponential map, then $T\phi^*f\in\Gamma(\Hat ST^*M)$ is constructed as above just by picking any generalized exponential map in the given equivalence class.
Choosing local coordinates $\{x^i\}$ on the base and $\{y^i\}$ on the fiber, we have explicit expressions
\begin{equation}\label{e:phiT}
\phi^i_x(y)= x^i+y^i+\frac12\phi^i_{x,jk}y^jy^k+\frac1{3!}\phi^i_{x,jkl}y^jy^ky^l+\cdots,
\end{equation}
and the class of $\phi$ is simply given by the collection of coefficients $\phi_{x,\bullet}$. One can easily see that the coefficients
$\phi^i_{x,jk}$ of the quadratic term transform as the components of a connection. We will refer to this as the connection in $\phi$.
Also explicitly we may compute
\begin{equation}\label{e:Tphistar}
T\phi_x^*f = f(x) + y^i\de_if(x)+\frac12y^jy^k(\de_j\de_kf(x)+\phi^i_{x,jk}\de_if(x))+\cdots.
\end{equation}

Above we have proved that sections of $\Hat ST^*M$ of the form $T\phi^*f$ satisfy (the Taylor expansion of) equation \eqref{e:dphif}. One can easily prove that the converse is also true. In fact,
one has even more. We may think of the Taylor expansion of the r.h.s.\  as an operator acting on the section $T\phi^*f$.  Actually, for every section $\sigma$ of $\Hat ST^*M$ one can define
a  section $R(\sigma)$ of $T^*M\otimes \Hat ST^*M$
by taking the Taylor expansion of  $-\dd_y\sigma\circ(\dd_y\phi)^{-1}\circ\dd_x\phi$. Notice that $R$ is  $C^\infty(M)$\ndash linear. As a consequence we have a connection\footnote{This is the Grothendieck connection in the presentation given by the choice of the formal exponential map $\phi$.}
\[
(X,\sigma)\in \Gamma(TM)\otimes\Gamma(\Hat S T^*M) \mapsto i_XR(\sigma)\in \Gamma(\Hat S T^*M)
\]
on $\Hat S T^*M$. One can check that this connection is flat.  
We can also regard $R$ as a one-form on $M$ taking values in the bundle $\End(\Hat ST^*M)$.
Also notice that $\Hat S T^*M$ is a bundle of algebras and that $R$ acts as a derivation;  so we can regard $R$ as a one-form on $M$ taking values in the bundle $\Der(\Hat ST^*M)$,
which is tantamount to saying the bundle of formal vertical vector fields $\Hat \frX(TM):= TM\otimes \Hat S T^*M$. Notice that the flatness of the connection may be expressed as the MC (Maurer--Cartan) equation
\begin{equation}\label{MCR}
\dd_xR+\frac12[R,R]=0,
\end{equation}
where $[\ ,\  ]$ is the Lie bracket of vector fields.
Finally, equation \eqref{e:dphif} may now be expressed by saying
that $\dd\sigma + R(\sigma)=0$ if $\sigma$ is of the form $T\phi^*f$ for some $f$. Below we will see that also the converse is true.

We first extend this connection to a differential $D$ on the complex of $\Hat S T^*M$\ndash valued differential
forms $\Gamma(\Lambda^\bullet T^*M\otimes \Hat S T^*M)$.\footnote{\label{f:L}Since
$\Gamma(\Lambda^\bullet T^*M\otimes \Hat S T^*M)$ is the algebra of functions on the formal graded manifold $\calM:=T[1]M\oplus T[0]M$, the differential $D$ gives $\calM$ the structure of
a differential graded manifold. In particular since $D$ vanishes on the body, we may linearize at each $x\in M$ and get an $L_\infty$\ndash algebra structure on
$T_xM[1]\oplus T_xM\oplus T_xM$.}
The main result is that the cohomology of $D$ is concentrated in degree zero and $H^0_D=T\phi^* C^\infty(M)$. This can be easily seen working in local coordinates again:
\[
R(\sigma)_i = \frac{\de\sigma}{\de y^k}\,\left(\left(\frac{\de\phi}{\de y}\right)^{-1}\right)^k_j\,\frac{\de\phi^j}{\de x^i}.
\]
Using \eqref{e:phiT} we get $R = \delta + R'$ with $\delta=-\dd x^i\,\frac\de{\de y^i}$ and $R'$ a one-form on the base taking value in the vector fields vanishing at $y=0$. Hence we have $D=\delta+D'$ with
\[
D' = \dd x^i\,\frac\de{\de x^i}+R'.
\]
Notice that $\delta$ is itself a differential and that it decreases the polynomial degree in $y$, whereas
the operator $D'$ does not decrease this degree. The fundamental remark is that the cohomology of $\delta$ consists of zero forms constant in $y$. This is easily shown
by introducing $\delta^*:=y^i\,\iota_{\frac\de{\de x^i}}$ and observing that $(\delta\delta^*+\delta^*\delta)\sigma=k\sigma$ if $\sigma$ is an $r$\ndash form of degree $s$ in $y$ and $r+s=k$.
By cohomological perturbation theory the cohomology of $D$ is isomorphic to the cohomology of $\delta$,
which is what we wanted to prove.

Finally, observe that, if $\sigma$ is a $D$\ndash closed section, we can immediately recover the function $f$ for which $\sigma=T\phi^*f$ simply by setting $y=0$,
$f(x)=\sigma_x(0)$, as follows from \eqref{e:Tphistar}.

We can now extend the whole story to other natural objects. Let $\calV(M)$ denote the multivector fields on $M$ (i.e., sections of $\Lambda TM$), $\Omega(M)$ the differential forms, $\calW^j(M):=\Gamma(S^jTM)$ and $\calO^j(M):=\Gamma(S^jT^*M)$.
We use similar symbols for formal vertical vector fields $\Hat\calV(TM):=\Gamma(\Lambda TM\otimes \Hat S T^*M)$ and formal vertical differential forms
$\Hat\Omega(TM):=\Gamma(\Lambda T^*M\otimes \Hat S T^*M)$. We have injective maps
\[
T\phi^*:=T(\phi_*)^{-1}\colon \calV(M)\to\Hat\calV(TM),\qquad
T\phi^*\colon\Omega(M)\to\Hat\Omega(TM).
\]
Similarly, we set $\Hat\calW^j(TM):=\Gamma(S^jTM\otimes \Hat S T^*M)$ and $\Hat\calO^j(TM):=\Gamma(S^jT^*M\otimes \Hat S T^*M)$ and get
\[
T\phi^*:=T(\phi_*)^{-1}\colon \calW^j(M)\to\Hat\calW(TM),\qquad
T\phi^*\colon\calO^j(M)\to\Hat\calO^j(TM).
\]

We can now let $R$ naturally act on $\Hat\calV(TM)$, $\Hat\Omega(TM)$, $\Hat\calW^j(TM)$ and $\Hat\calO^j(TM)$
by Lie derivative and hence get a differential $D$ on the corresponding complexes of differential forms.
Notice that $D$ respects the
Gerstenhaber algebra structure (by the vertical Schouten--Nijenhuis bracket) of $\Hat\calV(TM)$ and the differential complex structure (by the vertical differential) of $\Hat\Omega(TM)$, so that
these structures are induced in cohomology. By the same argument as above, we get that all these cohomologies are concentrated in degree zero with
$H^0_D(\Hat\calV(TM))=T\phi^*\calV(M)$, $H^0_D(\Hat\Omega(TM))=T\phi^*\Omega(M)$, $H^0_D(\Hat\calW^j(TM))=T\phi^*\calW^j(M)$,
and $H^0_D(\Hat\calO^j(TM)=T\phi^*\calO^j(M)$.
Notice in particular that a section 
is in the image of $T\phi^*$ if and only if
\[
\dd_x\sigma + L_R\sigma =0.
\]
In order to recover, in local coordinates,  the global object 
corresponding to a solution to the above equation, we should only observe that by assumption
$\dd_y\phi_x(0)=\id$, so that it is enough to evaluate the components of $\sigma$ at $y=0$ and to replace formally each $\dd y^i$ by $\dd x^i$ and each
$\frac\de{\de y^i}$ by $\frac\de{\de x^i}$. More explicitly, if $\sigma_x(y)=\sigma_{x; i_1,\dots,i_n}(y)\,\dd y^{i_1}\cdots\dd y^{i_n}$ is equal to $T\phi^*\omega$,
then
\[
\omega(x)=\sigma_{x; i_1,\dots,i_n}(0)\,\dd x^{i_1}\cdots\dd x^{i_n}.
\]
If on the other hand,
$\sigma_x(y)=\sigma_x^{i_1,\dots,i_n}(y)\,\frac\de{\de y^{i_1}}\cdots\frac\de{\de y^{i_n}}$ is equal to $T\phi^*(Y)$, then
\[
Y(x)=\sigma_x^{i_1,\dots,i_n}(0)\,\frac\de{\de x^{i_1}}\cdots\frac\de{\de x^{i_n}}.
\]

One can immediately extend these results to direct sums of the vector bundles above.
Notice that cohomology also commutes with direct limits. This implies that the cohomology of $\prod_j\Hat\calW^j(TM)$ is also concentrated in degree zero
and coincides with $T\phi^*\prod_j\calW^j(M)$. Now we have that $\prod_j\calW^j(M)=\Gamma(\Hat STM)$ whereas
$\prod_j\Hat\calW^j(TM)=\Gamma(\Hat S(TM\otimes T^*M))$. Similarly, we see that the cohomology with values in $\Hat S(T^*M\otimes T^*M)$
is concentrated in degree zero and coincides with $T\phi^*\Gamma(\Hat ST^*M)$. To summarize:
\[
H^\bullet_D(\Gamma(\Hat S(TM\otimes T^*M)))=
H^0_D(\Gamma(\Hat S(TM\otimes T^*M)))=
 T\phi^*\Gamma(\Hat STM)
\]
and
\[
H^\bullet_D(\Gamma(\Hat S(T^*M\otimes T^*M)))=
H^0_D(\Gamma(\Hat S(T^*M\otimes T^*M)))=
 T\phi^*\Gamma(\Hat ST^*M).
\]

\subsection{Gauge transformations}\label{ss:gauge}
We now wish to consider the effects of changing the choice of formal exponential map. Namely, let $\phi$ be a family of formal exponential maps depending on a parameter $t$
belonging to an open interval $I$. We may associate to this family a formal exponential map $\psi$ for the manifold $M\times I$ by
$\psi(x,t,y,\tau):=((\phi)_{x,t}(y),t+\tau)$, where $\tau$ denotes the tangent variable to $t$. We want to define the associated connection $\Tilde R$: on a section $\Tilde\sigma$ of
$\Hat S T^*(M\times I)$ we have, by definition,
\[
\Tilde R(\Tilde\sigma)=-(\dd_y\Tilde\sigma,\dd_\tau\Tilde\sigma)\circ
\begin{pmatrix}
(\dd_y\phi)^{-1} & 0\\
0 & 1
\end{pmatrix}\circ
\begin{pmatrix}
\dd_x\phi & \dot\phi\\
0 & 1
\end{pmatrix}.
\]
So we can write $\Tilde R = R + C\,\dd t + T$ with $R$ defined as before (but now $t$\ndash dependent),
\[
C(\Tilde\sigma) = -\dd_y\Tilde\sigma\circ(\dd_y\phi)^{-1}\circ\dot\phi
\]
and $ T=-\dd t\;\frac\de{\de \tau}$. We now spell out the MC equation for $\Tilde R$ observing that $\dd_x T=\dd_t T= 0$ and that $T$ commutes with both $R$ and $C$. The
$(2,0)$\ndash form component over $M\times I$ yields again the MC equation for $R$, whereas the the $(1,1)$\ndash component reads
\[
\dot R = \dd_x C + [R,C].
\]
Hence, under a change of formal exponential map, $R$ changes by a gauge transformation with generator the section $C$ of  $\Hat \frX(TM)$.

Finally, if $\sigma$ is a section in the image of $T\phi^*$, then by a simple computation one gets
\[
\dot\sigma = - L_C\sigma,
\]
which can be interpreted as the associated gauge transformation for sections.

\section{PSM in formal coordinates}\label{s-PSMformal}
The Poisson sigma model (PSM) \cite{I,SS} is a topological field theory with source a two-manifold $\Sigma$ and
target a Poisson manifold $M$. Before getting to the BV action for the PSM, we fix the notations and introduce the AKSZ formalism \cite{AKSZ}
(for a gentle introduction to it, especially suited to the PSM, see \cite{CF-AKSZ}). Let $\Map(T[1]\Sigma,T^*[1]M)$ be the infinite dimensional graded manifold of maps from $T[1]\Sigma$ to $T^*[1]M$. It fibers over $\Map(T[1]\Sigma,M)$. We denote by
$\sfX$ a
``point''
of $\Map(T[1]\Sigma,M)$ and by $\sfeta$ a ``point''
of the fiber. In local target coordinates, the super fields $\sfX$ and $\sfeta$ have simple expressions:
\begin{align*}
\sfX^i &= X^i + \eta_{+}^i + \beta_{+}^i,\\
\sfeta_i &= \beta_i + \eta_i + X^{+}_i,
\end{align*}
where we have ordered the terms in increasing order of form degree on $\Sigma$. The ghost number is $0$ for $X$ and $\eta$, $1$ for $\beta$, $-1$ for $\eta_+$ and $X^+$, and $-2$ for $\beta_+$. As unperturbed BV action one considers
\[
S_0:=\int_\Sigma \sfeta_i\,\dd\sfX^i.
\]
Notice that it satisfies the classical master equation (CME) $(S_0,S_0)=0$ if $\Sigma$ has no boundary or if appropriate boundary conditions are taken (which we assume throughout). Here $(\ ,\ )$ is the BV bracket corresponding to the
odd symplectic structure on the space of fields for which the superfield $\sfeta$ is the momentum conjugate to the superfield $\sfX$.
Formally one may also assume $\Delta S_0=0$ where $\Delta$ is the BV operator, so $S_0$ satisfies also the quantum master equation (QME).

To perturb this action, we pick a multivector field $Y$ on $M$. We may regard it as a function on $T^*[1]M$. We then define $S_Y$ as the integral over $T[1]\Sigma$ of the pullback of $Y$
by the evaluation map
\[
\ev\colon T[1]\Sigma\times \Map(T[1]\Sigma,T^*[1]M)\to T^*[1]M.
\]
Explicitly, for a $k$\ndash vector field $Y$, we have
\[
S_Y = \frac1{k!}\int_\Sigma Y^{i_1,\dots,i_k}(\sfX)\sfeta_{i_1}\dots\sfeta_{i_k}.
\]
This construction has several interesting properties. First, $(S_0,S_Y)=0$ (for $\de\Sigma=\emptyset$ or with appropriate boundary conditions). Second, for any two multivector fields $Y$ and $Y'$, we have $(S_Y,S_{Y'})=S_{[Y,Y']}$. The BV action for the PSM with target Poisson structure $\pi$ is recovered as $S=S_0+S_\pi$. Notice that by the above mentioned properties it satisfies the CME. As for the quantum master equation, we refer to \cite{CFeff}, where it is shown that one can assume $\Delta S_\pi=0$  if the Euler characteristic of $\Sigma$ is zero or
if $\pi$ is unimodular. In the latter case, one picks a volume form $v$ on $M$ such that $\Div_v\pi=0$ and defines $\Delta$ according to it.

We consider $S_\pi$ as a perturbation, so we expand the functional integral around the critical points of $S_0$. They consist of closed superfields. In particular, the component $X$ of $\sfX$ will be a constant map, say with image $x\in M$. Fluctuations will explore only a formal neighborhood of $x$ in $M$, so as in \cite[\S 6]{CF},
it makes sense to make the change of variables
\[
\sfX = \phi_x(\sfA),\quad\sfeta=\dd\phi_x(\sfA)^{*,-1}\sfB,
\]
where $\phi$ is a formal exponential map and the new superfields $(\sfA,\sfB)$ are in $\Map(T[1]\Sigma, T^*[1]T_xM)$. Notice that this change of variables
\[
\phi_x\colon \Map(T[1]\Sigma, T^*[1]T_xM)\to \Map(T[1]\Sigma,T^*[1]M)
\]
is a local  symplectomorphism and
 that
 \[
 T\phi_x^* S_0=\int_\Sigma\sfB_i\,\dd\sfA^i.
 \]
 The moduli space of vacua
 (i.e., the space of critical points modulo gauge transformations)
 is now
 $\calH_x:=H^\bullet(\Sigma)\otimes T_xM\oplus H^\bullet(\Sigma)\otimes T^*_xM[1]$.
 (Here $H^\bullet(\Sigma)$ is regarded as a graded vector space with its natural grading.)
 We should regard $\calH=\bigcup_x \calH_x$ as a vector bundle over $M$, but for the moment we concentrate on a single $x$. Later on we will also consider the remaining integration over vacua and, in particular, over $M$ (which actually shows up as the space of constant maps
 $\Sigma\to M$); see Section~\ref{s:pf}.

We may repeat the AKSZ construction on $\Map(T[1]\Sigma, T^*[1]T_xM)$. In particular, if $Y$ is a function of degree $k$ on $T^*[1]T_xM$ (i.e., a formal vertical
$k$\ndash vector field), we may construct
a functional
\[
S_Y= \frac1{k!}\int_\Sigma Y^{i_1,\dots,i_k}(\sfA)\sfB_{i_1}\dots\sfB_{i_k}.
\]
In particular, we have
\[
T\phi_x^*S_\pi=S_{T\phi_x^*\pi}.
\]
As a result, we have a solution $S_x:=T\phi_x^*S$ of the QME and may compute its partition function $Z_x$ (as a function on $\calH_x$) upon integrating over
 a Lagrangian submanifold $\calL$ of a complement of $\calH_x$ in $\Map(T[1]\Sigma, T^*[1]T_xM)$:
 \[
 Z_x:=\int_\calL\EE^{\frac\ii\hbar S_x}.
 \]
Notice \cite{M,CFeff} that there is an induced BV operator $\Delta$ on $\calH_x$ and that $Z_x$ satisfies $\Delta Z_x=0$. Moreover, upon changing the gauge fixing $\calL$, $Z_x$ changes
by a $\Delta$\ndash exact term. We wish to compare
the class of $Z_x$ 
with the globally defined partition function morally obtained by integrating in
$\Map(T[1]\Sigma,T^*[1]M)$. For this we have to understand the collection $\{Z_x\}_{x\in M}$ as a section
$\Hat Z\colon x\mapsto Z_x$
of $\Hat S\calH^*$ (we hide the dependency on $\hbar$ here)
and compute how it changes over $M$. Using all properties above and setting $\Hat S\colon x\mapsto S_x$,
we get\footnote{The choice of $\calL$ might be different for different $x$s, but for simplicity we assume it not to be the case.}
\[
\dd_x \Hat Z = \frac\ii\hbar\int_\calL \EE^{\frac\ii\hbar \Hat S}\, S_{\dd_x T\phi^*\pi}=
-\frac\ii\hbar\int_{\calL} \EE^{\frac\ii\hbar \Hat S}\,(S_R,\Hat S).
\]
Notice that this may also be rewritten as 
\[
\dd_x \Hat Z= -\Delta\int_{\calL} \EE^{\frac\ii\hbar \Hat S}\, S_R
\]
if we assume $\Delta S_R=0$. This is correct if $\Sigma$ has zero Euler characteristic or if $\Div_{T\phi^*v}R=0$. From the equation
$\dd_x T\phi^*v+L_RT\phi^*v=0$, we see that the latter condition is satisfied if and only if $\dd_x T\phi^* v=0$. Given a volume form $v$, it is always possible to find a formal exponential map $\phi$ satisfying this condition; actually, one can even get $T\phi_x^*v=\dd y^1\dots\dd y^d$ $\forall x$.

We can collect the above identities nicely if we define\footnote{For $\pi=0$, $\Tilde S$ is also the BV action for the $BF_\infty$\ndash theory \cite{M} with target the $L_\infty$ algebra of footnote~\ref{f:L}.
See also \cite{CG}.} 
\[
\Tilde S:= \Hat S + S_R.
\]
Notice that $\Tilde S$  is of total degree zero (the term $S_R$ has ghost number minus one  but is a one-form on $M$) and satisfies the modified CME
\[
\dd_x \Tilde S  + \frac12(\Tilde S,\Tilde S)=0
\]
and by assumption also $\Delta \Tilde S=0$ (so it satisfies a modified QME as well). We then define
\begin{equation} \label{SeffBVintegral}
\Tilde Z := \int_\calL\EE^{\frac\ii\hbar \Tilde S}
\end{equation}
as a nonhomogeneous differential form on $M$ taking values in $\calH$. It satisfies
\[
\dd_x \Tilde Z -\ii\hbar\Delta\Tilde Z =0.
\]
\begin{Rem}\label{r:changephi}
We are now also in a position to understand the change of $\Tilde Z$ under a change of the  formal exponential map. Using the results of subsection~\ref{ss:gauge}, we immediately
see that
\[
\dot{\Tilde Z} = (\dd_x-\ii\hbar\Delta)\int_\calL\EE^{\frac\ii\hbar \Tilde S}\,\frac\ii\hbar S_C,
\]
assuming $\Delta S_C=0$. The assumption is verified if $\Sigma$ has zero Euler characteristic or if we let $\phi$ vary only in the class of formal exponential maps that make $T\phi^*v$ constant.
Notice that the space of such formal exponential maps is connected. Therefore, the class of $\Tilde Z$
under these transformations is independent of all choices needed to compute it.
\end{Rem}

Finally, we consider the effective action $\Tilde S_\text{eff}$ defined by the identity $\Tilde Z = \EE^{\frac\ii\hbar\Tilde S_\text{eff}}$. It is a differential form taking values in $\Hat S\calH^*[[\hbar]]$ and
satisfies the modified QME
\begin{equation}\label{e:mQME}
\dd_x\Tilde S_\text{eff} +\frac12(\Tilde S_\text{eff},\Tilde S_\text{eff})-\ii\hbar\Delta\Tilde S_\text{eff}=0.
\end{equation}
This equation formally follows from the properties of BV integrals. In the case when $\Sigma$ has no boundary, it may be proved directly by considering the expansion of $\Tilde S_\text{eff}$ in Feynman diagrams
and applying the usual Stokes theorem techniques on integrals over configuration spaces, see \cite{K, CM}.
If $\Sigma$ has a boundary, additional terms corresponding to several points collapsing to the boundary together may appear and spoil \eqref{e:mQME}. From now on we therefore assume that $\Sigma$ is closed.

Equation  \eqref{e:mQME}
contains both information on the QME satisfied by the zero-form component $\Tilde S_\text{eff}^{(0)}$ and on its global properties. The equations are in general mixed: we do not simply get a flat connection with respect to which $\Tilde S_\text{eff}^{(0)}$ is covariantly constant. However, it is possible to find a modified quantum BV canonical transformation that produces a flat connection with respect to which the zero form part of the effective action is horizontal and hence global; we postpone this discussion to Section~\ref{s:global}. In the remaining of this Section, we concentrate on two special cases where the general theory is not needed.

The first special case is when $\Sigma$ is a torus.  In Section~\ref{s:caxial}, see Lemma~\ref{l-vanishing},  we will show that, in an appropriate gauge, there are no quantum corrections, so
\[
\Tilde S_\text{eff} = \Tilde S_\text{eff}^{(0)} + \Tilde S_\text{eff}^{(1)},
\]
where $\Tilde S_\text{eff}^{(0)}$ is the zero-form obtained by restricting $S_{T\phi^*\pi}$  to vacua and $\Tilde S_\text{eff}^{(1)}$ is the one-form obtained
by restricting $S_R$ to vacua. One can explicitly check, see Section~\ref{s:pf}, that $\Delta \Tilde S_\text{eff}^{(0)} =\Delta \Tilde S_\text{eff}^{(1)} =0$.
Hence the modified QME now simply yields the CME for $\Tilde S_\text{eff}^{(0)} $,
\[
(\Tilde S_\text{eff}^{(0)},\Tilde S_\text{eff}^{(0)})=0,
\]
the flatness condition for $ \Tilde S_\text{eff}^{(1)} $,
\[
\dd_x\Tilde S_\text{eff}^{(1)} +\frac12(\Tilde S_\text{eff}^{(1)},\Tilde S_\text{eff}^{(1)})=0,
\]
and the fact that $\Tilde S_\text{eff}^{(0)} $ is covariantly constant,
\[
\dd_x\Tilde S_\text{eff}^{(0)} +(\Tilde S_\text{eff}^{(1)},\Tilde S_\text{eff}^{(0)})=0.
\]
Now notice that $(\Tilde S_\text{eff}^{(1)} ,\ )$ is just the natural action of $R$ on the sections of $\calH$. Hence we can conclude that  $\Tilde S_\text{eff}^{(0)}$ is
just $T\phi^* (S|_\text{vacua})$.

The second special case is when $\pi$ is regular and unimodular (and $\Sigma$ is any two-manifold). Also in Section \ref{s:caxial} we will show that, upon choosing an appropriate formal exponential map,
$S_\text{eff}^{(0)}$ and $S_\text{eff}^{(1)}$ have no quantum corrections. Therefore, we may use the same reasoning as above and conclude that $\Tilde S_\text{eff}^{(0)}$ is
just $T\phi^* (S|_\text{vacua})$.

A final important remark is that in the two cases above
 the effective action 
depends polynomially on all vacua but, possibly, for those related to the $X$ field; therefore,
$S_\text{eff}^{(0)}$ is a section of $S\Tilde\calH^*\otimes\Hat ST^*M$, where
$\Tilde\calH_x=H^{>0}(\Sigma)\otimes T_xM\oplus H^\bullet(\Sigma)\otimes T^*_xM[1]$ is the moduli space of vacua excluding those for $X$. The corresponding global effective action $S|_\text{vacua}$ will then be a section of $S\Tilde\calH^*$, i.e., a function on the vector bundle $\Tilde\calH$ (polynomial in the fibers).
Notice that this vector bundle is diffeomorphic, by choosing a connection (e.g., the one contained in the choice of $\phi$), to the natural global definition of the moduli space of vacua as presented, e.g., in \cite{BMZ}.

\section{Some computations of the effective action}\label{s:caxial}
In this Section we discuss the perturbative computation for the effective action and show that it has no quantum corrections in two important cases.
\subsection{Factorization of Feynman graphs}
Consider the effective action $ \Tilde S_\text{eff}$ defined in \eqref{SeffBVintegral} by the identity $\Tilde Z = \EE^{\frac\ii\hbar\Tilde S_\text{eff}}$.
Here the Lagrangian subspace $\LL$ in the complement of $\HH_x$ inside the space of fields
\be \FF_x = \Maps(T[1]\Sigma,T^*[1]T_x M)\cong \Omega^\bt(\Sigma)\otimes (T_x M \oplus T^*_x[1] M) \label{F_x}\ee
accounts for the gauge fixing. Let $\LL$ have the factorized form
\be \LL= \LL_K \otimes (T_x M \oplus T^*_x[1] M) \label{L factorization}\ee
where $\LL_K\subset \Omega^\bt(\Sigma)$ is defined as
$$\LL_K=\ker\PP\cap \ker K$$
with $\PP$  the projector from differential forms on $\Sigma$ to (the chosen representatives of) de Rham cohomology of $\Sigma$
and
$$K\colon \Omega^\bt(\Sigma)\ra \Omega^{\bt-1}(\Sigma)$$
a linear operator satisfying
\be \dd K+ K\dd = \id- \PP, \qquad \PP K = K \PP=0,\qquad K^T=K, \qquad K^2=0
\label{K properties}
\ee
(i.e., $K$ is the chain homotopy between identity and the projection to cohomology, also known as a parametrix). The transpose is w.r.t.\ the Poincar\'e pairing on forms $\int_\Sigma \bt\wedge\bt$. We assume the operator $K$ (which now determines the gauge fixing) to be an integral operator
with a distributional integral kernel $\omega\in\Omega^1(\Sigma\times \Sigma)$ -- the propagator. An explicit construction may be done along the same lines as in \cite{BC,C}.
Let us introduce a basis $\{\chi_\alpha\}$ in the cohomology space $H^\bt(\Sigma)$; denote the matrix of
the Poincar\'e pairing by
$\Pi_{\alpha\beta}=\int_\Sigma \chi_\alpha\wedge \chi_\beta$. In terms of $\omega$, properties (\ref{K properties}) read:
\begin{enumerate}
\item $\dd\omega = \delta_\mr{diag}-\sum_{\alpha,\beta} (\Pi^{-1})^{\alpha\beta}\chi_\alpha\otimes \chi_\beta$,
where $\delta_\mr{diag}$ is the delta-form supported on the diagonal of $\Sigma\times\Sigma$;
\item $\int_{\Sigma_{(1)}}\omega\,\pi_1^*\chi_\alpha=\int_{\Sigma_{(2)}}\omega\,\pi_2^*\chi_\alpha=0,\ \forall\alpha$,
where  $\Sigma_{(1)}$ and
$\Sigma_{(2)}$ denote the two factors of $\Sigma\times\Sigma$, and
$\pi_1$ and $\pi_2$ are the two projections from $\Sigma\times\Sigma$ to its factors;
\item $t^*\omega=\omega$, where $t\colon \Sigma\times\Sigma \ra \Sigma\times\Sigma$ is the map swapping the two copies of $\Sigma$;
\item $\int_{\Sigma_{(2)}}\pi_{12}^*\omega\,\pi_{23}^*\omega=0$, where $\Sigma_{(2)}$ denotes the middle factor
in $\Sigma\times\Sigma\times\Sigma$, whereas $\pi_{12}$ and $\pi_{23}$ are the projections from
$\Sigma\times\Sigma\times\Sigma$ to the first two and last two factors, respectively.
\end{enumerate}
Notice
that the restriction of $\omega$ to the configuration space
$C_2^0(\Sigma):=\{(u,v)\in\Sigma:u\not=v\}$ is smooth and it extends to the
Fulton--MacPherson--Axelrod--Singer compactification as a smooth form.
In \cite{CM}  it is shown how to implement the property $\PP K = K \PP=0$
on the propagator. Once this is done, the propagator will also satisfy the property
$K^2=0$. This is proved exactly as in \cite[Lemma 10]{CFeff}.

The perturbation expansion for the effective action (\ref{SeffBVintegral}) has the form
\be S^\mr{eff}(\mathsf{A}_\mr{z.m.},\mathsf{B}_\mr{z.m.};\hbar)=\sum_\Gamma \frac{(i\hbar)^{l(\Gamma)}}{|\mr{Aut}(\Gamma)|}\;\; W^\mr{target}_\Gamma (\mathsf{A}_\mr{z.m.},\mathsf{B}_\mr{z.m.})\cdot W_\Gamma^\mr{source}
\label{S_eff pert expansion}\ee
where the sum is over connected oriented graphs $\Gamma$ with leaves\footnote{A leaf for us is a loose half-edge, i.e., one not connected to another half-edge to form an edge.} decorated by basis cohomology classes $\{\chi_\alpha\}$; $l(\Gamma)$ stands for the number of loops, $|\mr{Aut}(\Gamma)|$ is the number of graph automorphisms; $\{\mathsf{A}_\mr{z.m.},\mathsf{B}_\mr{z.m.}\}=\{\mathsf{A}^{\alpha i}, \mathsf{B}^\alpha_i\}$ are the coordinates on the moduli space of vacua
$\HH_x$.

The ``target part'' $W^\mr{target}_\Gamma$ of the contribution of a graph to $\Tilde S_\text{eff}$ is a homogeneous polynomial function on $\HH_x$ of degree equal to the number of leaves, computed using the following set of rules:
\begin{enumerate}
\item to an incoming leaf of $\Gamma$ decorated by $\chi_\alpha$ one associates $\mathsf{A}^{\alpha i}$
\item to an outgoing leaf decorated by $\chi_\alpha$ one associates $\mathsf{B}^\alpha_i$
\item to a vertex with $m$ inputs and $n$ outputs one associates the expression
$$\partial_{i_1}\cdots\partial_{i_m} Y^{j_1\cdots j_n}$$
-- the $m$-th derivative of $n$-vector\footnote{In the standard setup for the Poisson sigma model, only the perturbation by Poisson bivector field $Y=\pi^{ij}\partial_i\wedge\partial_j$ is present in the action, hence all vertices have to have exactly two outputs, otherwise the graph does not contribute. In the present case, we also have the vector field $R$.} contribution to the action $S$.
\item for every edge contract the dummy Latin indices for the two constituent half-edges.
\end{enumerate}
The result of contraction is a polynomial function on $\HH_x$.

The ``source'' (or ``de Rham'') part $W^\mr{source}_\Gamma$ is a number defined as
\be W^\mr{source}_\Gamma=\int_{\Sigma^{\times V(\Gamma)}}\left(\prod_{\mr{edges}\; (h_\mr{in},h_\mr{out})}\pi^*_{v(h_\mr{in}),v(h_\mr{out})}\omega \right)\cdot \left(\prod_{\mr{leaves}\; l}\pi^*_{v(l)} \chi_{\alpha_l}\right) \ee
where $V(\Gamma)$ is the number of vertices, $\pi_v: \Sigma^{\times V(\Gamma)}\ra \Sigma$ is the projection to $v$-th copy of $\Sigma$, $\pi_{u,v}:\Sigma^{\times V(\Gamma)}\ra \Sigma\times \Sigma$ is the projection to $u$-th and $v$-th copies of $\Sigma$; $v(h)$ is the vertex incident to the half-edge $h$.
Notice that all these integrals converge. The usual way to show this is to observe that the integrals
are actually defined on configuration spaces (i.e., the complements of all diagonals in the Cartesian products
of copies of $\Sigma$) and that the propagators $\omega$ extend to their compactifications.

\begin{Rem} The factorization into source- and target contributions for Feynman diagrams in the expansion (\ref{S_eff pert expansion}) is due to the factorization of the space of fields (\ref{F_x}) and to the fact that our ansatz for the gauge fixing (\ref{L factorization}) is compatible with this factorization.
\end{Rem}

\begin{Rem} The orientation of $\Gamma$ is irrelevant for the source parts $W^\mr{source}_\Gamma$.
\end{Rem}

\subsection{Regular Poisson structures}
If $\pi$ is nondegenerate, it is always possible to find a formal exponential map $\phi$ such that $T\phi^*\pi$ is constant (in the $y$ variables).
One simply has to go to formal Darboux coordinates.
Notice, moreover, that $\Div_v\pi=0$ if for $v$ one chooses $v$ to be the Liouville volume form $\omega^k/k!$,
$k=\dim M/2$. It then follows that $T\phi^*v$ is also constant
and that $\Div_{T\phi^* v}R=0$.
A slight generalization occurs when $\pi$ is regular (i.e., its kernel has constant rank) and unimodular (notice that this is not guaranteed if $\pi$ is degenerate \cite{W}). After choosing $v$ such that $\Div_v\pi=0$, it is again possible to find a formal exponential map $\phi$ such that $T\phi^*\pi$ and $T\phi^*v$ are both constant and hence
$\Div_{T\phi^* v}R=0$.

In the perturbative expansion, we may thus assume that we have a bivalent vertex, corresponding to $T\phi^*\pi$, with no incoming arrows. If one of the outgoing arrows is replaced by a vacuum mode (i.e., a cohomology class),
the result is zero by the property $\PP K = K \PP=0$, otherwise it is zero by the property $K^2=0$.
As a result, every graph containing a $T\phi^*\pi$\ndash vertex will vanish, apart from the one with both outgoing arrows evaluated on vacua. As a consequence $S_\text{eff}^{(0)}$ and $S_\text{eff}^{(1)}$ have no quantum corrections.

\subsection{Axial gauge on the torus $\Sigma=\T^2:=S^1\times S^1$}\label{s:axial}
In the case of a torus, differential forms have a bigrading with respect to the two circles.
One may choose the axial gauge\footnote{The axial gauge for topological field theories was originally proposed in the context of Chern--Simons theory in \cite{FK}.} 
by setting the superfields to vanish if they have nonzero degree with respect to the first circle. Na\" ively this implies the propagator to be the product
of a propagator for the de Rham differential on the first circle and the identity operator on the second circle (just plug in the gauge fixed fields into in the unperturbed action to realize this). This argument however does not take vacua into account nor the fact that the axial gauge fixing does not fix all the gauge freedom. In fact, one can prove that the
propagator in the axial gauge has one additional term, see \eqref{K torus axial} below.

To start with a rigorous construction of the propagator, observe that
differential forms on a circle admit the Hodge decomposition
$$\Omega^\bt(S^1)=\underbrace{\Omega^\bt_\mr{Harm}(S^1)}_{\mr{Span}(1,d\tau)}\oplus \underbrace{\tilde\Omega^0(S^1)}_{\{f(\tau)\, | \, \int_{S^1} f(\tau)d\tau=0\}}\oplus \underbrace{\tilde\Omega^1(S^1)}_{\{g(\tau)d\tau\, | \, \int_{S^1}g(\tau)d\tau=0\}}$$
(In our convention the coordinate $\tau$ on the circle runs from $0$ to $1$).
The associated chain homotopy operator is
$$K_{S^1}: g(\tau)d\tau \mapsto \int_{S^1}\omega_{S^1}(\tau,\tau')g(\tau')d\tau'$$
with the integral kernel
$$\omega_{S^1}(\tau,\tau')=\theta(\tau-\tau')-\tau+\tau'-\frac{1}{2}$$
Projection to harmonic forms on the circle (representatives of cohomology) is
$$\PP_{S^1}: f(\tau)+g(\tau)d\tau \mapsto \int_{S^1}(d\tau'-d\tau)\wedge (f(\tau')+g(\tau')d\tau')$$

For the torus we may decompose the de Rham complex in the following way:
\begin{multline}
\Omega^\bt(S^1\times S^1)=\Omega^\bt(S^1)\hat\otimes \Omega^\bt(S^1) = \\
=\underbrace{\Omega^\bt_\mr{Harm}(S^1)\otimes\Omega^\bt_\mr{Harm}(S^1)}_{\cong H^\bt(S^1\times S^1)} \oplus \underbrace{\tilde\Omega^0(S^1)\hat\otimes \Omega^\bt \oplus
\Omega^\bt_\mr{Harm}(S^1)\otimes \tilde\Omega^0(S^1)}_{\LL_K} \oplus \\
\oplus\tilde\Omega^1(S^1)\hat\otimes \Omega^\bt \oplus
\Omega^\bt_\mr{Harm}(S^1)\otimes \tilde\Omega^1(S^1)
\end{multline}
The associated chain homotopy operator is
\be K=\underbrace{K_{S^1}\otimes \id_{S^1}}_{K_I}+ \underbrace{\PP_{S^1}\otimes K_{S^1}}_{K_{II}} \colon \Omega^\bt(S^1\times S^1)\ra \Omega^{\bt-1}(S^1\times S^1)  \label{K torus axial}\ee
Its integral kernel (the propagator) is
\begin{multline}
\omega=\underbrace{\left(\theta(\sigma-\sigma')-\sigma+\sigma'-\frac{1}{2}\right)\cdot \delta(\tau-\tau')\cdot (d\tau'-d\tau)}_{\omega_{I}}+\\
+\underbrace{(d\sigma'-d\sigma)\cdot \left(\theta(\tau-\tau')-\tau+\tau'-\frac{1}{2}\right)}_{\omega_{II}}
\label{propagator torus axial}
\end{multline}
where we denote by $\sigma,\tau\in \mathbb{R}/\mathbb{Z}$ the coordinates on the first and the second circles, respectively.

\begin{Rem} The chain homotopy (\ref{K torus axial}) arises from the composition of two 
quasi-isomorphisms:{\tiny
$$
\begin{CD}
\Omega^\bt(S^1)\hat\otimes\Omega^\bt(S^1) @= \Omega^\bt_\mr{Harm}(S^1)\otimes \Omega^\bt(S^1)\oplus \tilde\Omega^0(S^1)\hat\otimes \Omega^\bt(S^1) \oplus \tilde\Omega^1(S^1)\hat\otimes \Omega^\bt(S^1) \\
@VVV \\
\Omega^\bt_\mr{Harm}(S^1)\otimes \Omega^\bt(S^1) @= \Omega^\bt_\mr{Harm}\otimes \Omega^\bt_\mr{Harm} \oplus \Omega^\bt_\mr{Harm}(S^1)\otimes \tilde\Omega^0(S^1)\oplus  \Omega^\bt_\mr{Harm}(S^1)\otimes \tilde\Omega^1(S^1)\\
@VVV \\
\Omega^\bt_\mr{Harm}(S^1)\otimes \Omega^\bt_\mr{Harm}(S^1)
\end{CD}
$$}
i.e., we first contract the first circle to cohomology, then the second one.
\end{Rem}

\subsection{Vanishing of quantum corrections}
\begin{Lem}\label{l-vanishing}
For the Poisson sigma model in the axial gauge on torus, the source parts $W^\mr{source}_\Gamma$ vanish for all connected graphs $\Gamma$ except for trees with one vertex (``corollas'').
\end{Lem}

\textbf{Proof.} Let us introduce the basis in cohomology of the torus:
$$\chi_{(0,0)}=1,\quad \chi_{1,0}=d\sigma,\quad \chi_{(0,1)}=d\tau, \quad \chi_{(1,1)}=d\sigma\wedge d\tau$$
 Define a decoration $c$ of $\Gamma$ as an assignment of bidegree
 \[
 c(h)\in\{(0,0),(1,0),(0,1),(1,1)\}
 \]
 to each half-edge $h$ of $\Gamma$ (so that on leaves the bidegree coincides with the prescribed leaf decoration $\alpha$) together with an assignment of an index $c(e)\in \{I, II\}$ to each edge $e$. Define the source part for a decorated graph $\Gamma$ as
 \be
 W^\mr{source}_{\Gamma,c}=\int_{\Sigma^{\times V(\Gamma)}}\left(\prod_{\mr{edges}\; e=(h_\mr{in},h_\mr{out})}\pi^*_{v(h_\mr{in}),v(h_\mr{out})}\omega_{c(e)}|_c \right)\cdot \left(\prod_{\mr{leaves}\; l}\pi^*_{v(l)} \chi_{\alpha_l}\right) \label{W^source_c}
 \ee
where the $\omega|_c$ symbol means the component of the propagator (as an element of $\Omega^\bt(S^1\times S^1)\hat\otimes \Omega^\bt(S^1\times S^1)$) of de Rham bidegrees $c(h_{in}), c(h_{out})$ where $h_{in},h_{out}$ are the constituent half-edges of the edge; $\omega_{c(e)}$ is one of the two pieces of propagator, $\omega_I$ or $\omega_{II}$, as defined in (\ref{propagator torus axial}). Then we have
$$W^\mr{source}_\Gamma = \sum_{\mr{decorations}\; c}W^\mr{source}_{\Gamma,c}$$
The source part $W^\mr{source}_{\Gamma,c}$ vanishes automatically unless the following conditions are satisfied simultaneously:
\begin{enumerate}[(i)]
\item At every vertex there is exactly one incident half-edge decorated by $(1,\bullet)$, all others are $(0,\bullet)$. \label{decoration restrictions i}
\item At every vertex there is exactly one incident half-edge decorated by $(\bullet,1)$, all others are $(\bullet,0)$. (This half-edge may be the same as in (\ref{decoration restrictions i})).
    \label{decoration restrictions ii}
\item Compatibility between edge decorations and half-edge decorations: for any edge $e=(h_1,h_2)$ we have
\begin{eqnarray*}
(c(e)=I) \Longrightarrow \left[\begin{array}{l}
c(h_1)=(0,0), c(h_2)=(0,1) \qquad\mbox{or}\\ c(h_1)=(0,1), c(h_2)=(0,0)
\end{array}\right. \\
(c(e)=II) \Longrightarrow \left[\begin{array}{l}
c(h_1)=(0,0), c(h_2)=(1,0) \qquad\mbox{or}\\ c(h_1)=(1,0), c(h_2)=(0,0)
\end{array}\right.
\end{eqnarray*}
\label{decoration restrictions iii}
\item Number of edges decorated as $I$ adjacent to any given vertex should be different from one.
\label{decoration restrictions iv}
\item If a vertex has no adjacent $I$-edges, then the number of adjacent $II$-edges should be different from one.
\label{decoration restrictions v}
\end{enumerate}
Requirements (\ref{decoration restrictions i},\ref{decoration restrictions ii}) follow directly from degree counting in (\ref{W^source_c}); (\ref{decoration restrictions iii}) follows from the formula for propagator (\ref{propagator torus axial}); (\ref{decoration restrictions iv},\ref{decoration restrictions v}) follow from the property $K_{S^1}\PP_{S^1}=0$ and from the fact that harmonic forms on a circle are closed under wedge multiplication.

Fix some decoration $c$ of $\Gamma$ satisfying (\ref{decoration restrictions i}--\ref{decoration restrictions v}). Consider the subgraph $\Gamma_I$ of $\Gamma$ obtained by deleting all $II$-edges in $\Gamma$; $\Gamma_I$ may be disconnected. Let $\Gamma_I=\sqcup_a \Gamma_I^a$ where $\Gamma_I^a$ are the connected components of $\Gamma_I$.
Due to (\ref{decoration restrictions ii}), the number of vertices $V_I^a$ of $\Gamma_I^a$ is equal to the number of $(0,1)$-half-edges in $\Gamma_I^a$ which is in turn greater or equal to the number of edges $E_I^a$ due to (\ref{decoration restrictions iii}). Hence the Euler characteristic of $\Gamma_I^a$ non-negative:
$V_\Gamma^a-E_\Gamma^a\geq 0 $.
Therefore $\Gamma_I^a$ is either a tree or a 1-loop graph. Next, property (\ref{decoration restrictions iv}) shows that $\Gamma_I^a$ has to be a wheel graph, with arbitrary number of leaves attached at vertices, or a corolla.
On the other hand, if $\Gamma_I$ contains a wheel then the corresponding source part vanishes:
\begin{multline}
W^\mr{source}_{\Gamma,c} = \int_{(S^1\times S^1)^{\times V}} (d\tau_1-d\tau_2)\delta(\tau_2-\tau_1)\wedge\cdots \wedge
(d\tau_n-d\tau_1)\delta(\tau_1-\tau_n)\wedge F  =\\
\int_{(S^1\times S^1)^{\times V}} \underbrace{(d\tau_1\wedge d\tau_2\wedge\cdots \wedge d\tau_n+ (-1)^n d\tau_2\wedge\cdots d\tau_n\wedge d\tau_1)}_{=0} \wedge\\
\wedge\delta(\tau_2-\tau_1)\cdots\delta(\tau_1-\tau_n)\wedge F
=0
\label{vanishing argument}
\end{multline}
where $n$ is the length of the wheel and $F\in \Omega^\bt((S^1\times S^1)^{\times V})$ is some differential form.

\begin{Rem}
Argument (\ref{vanishing argument}) has the fault that the integrand is singular and the result is $0\cdot \delta(0)$. This can be remedied by regularizing the propagator $\omega$, e.g., by changing $\delta(\tau-\tau')$ in (\ref{propagator torus axial}) to a smeared delta-function. Notice that the
source parts of all diagrams except corollas still vanish exactly: in this vanishing argument the chain homotopy equation is never used; we only use
the de~Rham bigrading properties, $PK=0$ and the fact that harmonic forms on a circle are closed under multiplication.
\end{Rem}


Thus we have shown that $W^\mr{source}_{\Gamma,c}$ vanishes unless $\Gamma_I$ is a collection of corollas (i.e. there are no $I$-edges).

Now fix a decoration $c$ satisfying (\ref{decoration restrictions i}--\ref{decoration restrictions v}) with $c(e)=II$ for all edges. Repeating the Euler characteristic argument as above (using properties (\ref{decoration restrictions i},\ref{decoration restrictions iii})), we show that $\Gamma$ has to be either a tree or a 1-loop graph and using property (\ref{decoration restrictions v}) we show that it has to be either a wheel or a corolla. If it is a wheel then
\begin{multline}
W^\mr{source}_{\Gamma,c} = \int_{(S^1\times S^1)^{\times V}} (d\sigma_1-d\sigma_2)\wedge\cdots\wedge (d\sigma_V-d\sigma_1) \wedge F=\\
=\int_{(S^1\times S^1)^{\times V}} \underbrace{(d\sigma_1\wedge\cdots\wedge d\sigma_V + (-1)^V d\sigma_2\wedge\cdots\wedge d\sigma_V \wedge d\sigma_1)}_{=0}\wedge F=\\
=0
\end{multline}
Therefore $W^\mr{source}_{\Gamma,c}$ vanishes for any decoration $c$ unless $\Gamma$ is a corolla. This concludes the proof of the Lemma.
$\Box$

An immediate consequence of the Lemma is that the effective action $S^\mr{eff}_x$ is just the restriction of the action $S_x$ to vacua: there are no quantum corrections.

\section{The partition function on the torus}\label{s:pf}
Let $\Seff$ be the global effective action on the moduli space of vacua for
the torus $\T^2$. In Lemma~\ref{l-vanishing}, we have shown that it has no quantum
corrections. The moduli space of vacua can be viewed as $\Map(\bbR^2[1],T^*[1]M)$ and in \cite{BMZ} it has been remarked that the action restricted to vacua
is the AKSZ action for this mapping space.
In local coordinates the
superfields are
\begin{equation}
\label{reduced_superfields}
 \superx^\mu = x^\mu + e^1 \eta^{+\mu}_1 + e^2 \eta^{+\mu}_2 - s b^{+\mu} \;,
  ~~~~ \supereta_\nu = b_\nu + e^1 \eta_{\nu 1} + e^2 \eta_{\nu 2} + s x^+_\nu\;,
\end{equation}
where $s=e^1 e^2$ is the generator of $H^2_{dR}(\T^2)$ normalized to $\int\limits_{\bbR^2[1]} ds\ s =1$. If $\pi$ is the Poisson bivector field on $M$,
then 
\begin{equation}
\label{solution_master_equation}
\Seff = \frac12\int_{\bbR^2[1]} ds \, \pi^{\mu\nu}(\superx) \supereta_\mu\supereta_\nu \;.
\end{equation}

There exists a canonical Berezinian given by the coordinate volume form
\begin{equation}
\label{torus_berezinian}
\nu = dx\cdots dx^+\cdots db \cdots db^+ \cdots d\eta_i\cdots d\eta^+_i\cdots \;.
\end{equation}
If we denote with $\Delta$ the corresponding Laplacian, the AKSZ action satisfies
$$\Delta \EE^{\frac\ii\hbar\Seff}=0$$
and defines a class in $\Delta$\ndash cohomology.

\subsection{K\"ahler gauge fixing and Euler class}\label{ss:kahler}
Now let $\pi$
be symplectic such that $\pi^{-1}$ is the K\"ahler form of the hermitian structure $(J,g)$. In the complex coordinates $\{x^i\}$ of $M$ we have
$\pi^{i\J} = i g^{i\J}$. Let us fix a complex structure on $\T^2$ defined by $z=\theta^1+\tau \theta^2$, for $\tau=\tau_1 + i \tau_2$ and $\tau_2>0$.
Let
$$
\eta^{+\mu}_z = (\eta_2^{+\mu}-\bar\tau \eta^{+\mu}_1)/2i \tau_2\;,~~~~~\eta_{z\mu} = (\eta_{2\mu}-\bar\tau \eta_{1\mu})/2i \tau_2\;.
$$
Let ${\calL}_{\varepsilon,\tau}$ be the
following Lagrangian submanifold  of $\Map(\bbR^2[1],T^*[1]M)$:
\begin{equation}
\label{Kahler_gauge_fixing}
\eta^{+i}_z = \eta^{+\I}_{\z} =\eta_{zi} = \eta_{\z \I} = x^+ = b^+ = 0  ~~.
\end{equation}
Let us define
$$
p_k= \eta_{\z k} + \Gamma^{j}_{ki} \eta^{+i}_{\z} b_j \;,
$$
where $\Gamma$ are the Christoffel symbols of the Levi-Civita connection. All fiber coordinates $b,p,\eta^+$ transform tensorially with respect
to a transformation of coordinates on $M$ so that
${\calL}_{\varepsilon,\tau}=(T^*[1]+T^*M +T[-1])M$. After a straightforward computation we get
\begin{eqnarray*}
\Seff&=&\tau_2 \left (R^j_{s\bar l k}g^{s\bar r} \eta^{+k}_{\z} \eta^{+\bar l}_z b_j b_{\bar r} +  g^{i\bar j} p_i p_{\bar j}\right )\cr
 & = & \frac12 \tau_2 \left( R^{\mu\lambda} b_\mu b_\lambda + g^{\mu\nu} p_\mu p_\nu\right)\;.
\end{eqnarray*}
The induced Berezinian on ${\calL}_{\varepsilon,\tau}$ reads
$$
\sqrt{\nu} = \frac{dx^\mu\cdots db_\mu \cdots dp_\mu\cdots  d\eta^{+\mu}\cdots}{(2\pi)^m},
$$
with $m=\dim M$.\footnote{We choose here the standard convention that the measure for a pair of even conjugate coordinates $p,q$ is $dp\,dq/(2\pi\ii\hbar)$,
whereas the measure for a pair of odd conjugate coordinates $b,\eta^+$ is $\ii\hbar\, db\,d\eta^+$. This is consistent with the standard normalization
\[
\int\EE^{\frac\ii\hbar pq} \frac{dp\,dq}{2\pi\ii\hbar} = \int \EE^{\frac\ii\hbar b\eta^+}\ii\hbar\, db\,d\eta^+=1.
\]}
If we perform the fiberwise integration with respect to the fibration ${\calL}_{\varepsilon,\tau}\rightarrow T[-1]M$, we get
\begin{multline*}
\frac1{(2\pi)^m}\int dp_\mu\cdots db_\mu\cdots  \EE^{\frac\ii\hbar\Seff} = \\
=\frac1{(2\pi)^m}\frac1{(\ii\hbar)^{\frac m2}}
\int db_\mu\cdots \frac{1}{\det(\tau_2 g^{\mu\nu})^{1/2}} \EE^{\frac\ii{2\hbar}\tau_2 R^{\mu\nu}b_\mu b_\nu} =\\
=\frac1{(2\pi)^\frac m2}\int db'_\mu\cdots \frac{1}{\det(g^{\mu\nu})^{1/2}}
\EE^{-\frac12R^{\mu\nu}b'_\mu b'_\nu} = \\
=\frac1{(2\pi)^\frac m2}\sqrt g\,\Pf(R)
\in C^\infty(T[-1]M)=\Omega M,
\end{multline*}
which, by the Chern--Gauss-Bonnet theorem, is a representative of the Euler class.
Notice that $\hbar$ and $\tau_2$ disappear in the final formula. (The main reason for this is that
scaling the $b$ and $p$ variables by the same factor preserves the Berezinian since the former are odd and the latter are even variables.)
Finally, we can integrate over $M$ getting
\[
Z = \chi(M),
\]
the Euler characteristic of $M$. (Actually, by the argument in the Introduction that the partition function should be the Euler characteristic, we might in reverse think of this result as one more physical proof of the Chern--Gauss--Bonnet theorem, in the case of K\"ahler manifolds.)

\begin{Rem}
In this Section we have assumed the existence of a K\"ahler structure on $M$. We expect the above results to
hold if we just use an almost K\"ahler structure, but computations become much more involved.
\end{Rem}
\begin{Rem}
Another possible gauge fixing consists in setting all $+$ variables to zero. The effective action
then reduces to $\pi^{\mu\nu}(x)\eta_{\mu1}\eta_{\nu2}$ and is independent of the $b$ variables.
If $M$ is compact, the integrals over the $\eta$ and $x$ variables is finite (and proportional to the
symplectic volume of $M$); because of the $b$\ndash integration, the partition function then vanishes.
If $M$ is not compact, the partition function is ambiguous and of the form $0\cdot\infty$. This gauge
fixing is then in general not equivalent to the K\"ahler one used above. From the considerations in the Introduction, the K\"ahler gauge fixing is the one compatible with the Hamiltonian interpretation of the theory.
\end{Rem}

\subsection{Regularized effective action}\label{ss-rea}
We now show that the symmetries of $\Seff$ induce a regularization which allows one to compute the partition function for every Poisson structure and to show that, independently
of the Poisson structure, one gets the Euler characteristic of the target.\footnote{We thank T. Johnson-Freyd for pointing out this approach.}

The main remark is that the effective action and the symplectic form are invariant under the action of the Lie algebra of divergenceless vector fields of $\bbR^2[1]$
on the moduli space of vacua. This Lie algebra is spanned by the vector fields
$\frac\de{\de e^1}$,  $\frac\de{\de e^2}$, $e^2\frac\de{\de e^1}$, $e^1\frac\de{\de e^2}$ and $e^1\frac\de{\de e^1}-e^2\frac\de{\de e^2}$.
The fifth vector field is generated by the previous ones and we are not going to need it in the following.
We lift the first four vector fields first to $\Map(\bbR^2[1],M)$ and next to its cotangent bundle shifted by one. We will denote the resulting vector fields by
$\delta_1$ , $\delta_2$, $K_1$ and $K_2$, respectively.
Since they have degree $-1$ or $0$
and the symplectic form has degree $-1$, they are also automatically Hamiltonian with uniquely defined Hamiltonian functions $\tau_1$ and $\tau_2$ (of degree $-2$), and $\rho_1$ and $\rho_2$ (of degree $-1$).
The Lie algebra relations translate into the Poisson bracket relations
\begin{equation}\label{eone}
(\tau_2,\rho_1)=\tau_1,\quad (\tau_1,\rho_2)=\tau_2,\quad
(\tau_1,\rho_1)=0,\quad(\tau_2,\rho_2)=0.
\end{equation}
Also notice that we have
\begin{equation}\label{etwo}
K_1\circ\delta_1=K_2\circ\delta_2=0,
\end{equation}
which implies that $\rho_i$ Poisson commutes with every $\delta_i$\ndash exact function.
Finally, since we started with divergenceless vector fields, we get
\begin{equation}\label{ethree}
\Delta\tau_1=\Delta\tau_2=\Delta\rho_1=\Delta\rho_2=0.
\end{equation}

\begin{Rem}
Even though we do not need the explicit form of these vector fields and their Hamiltonian functions, we give them for completeness of our presentation.
From the defining formulae
$\delta_i\superx= \frac{\de\superx}{\de e^i}$, $\delta_i\supereta = \frac{\de\supereta}{\de e^i}$
$K_1\superx=e^2\frac{\de\superx}{\de e^1}$, $K_1\supereta=e^2\frac{\de\supereta}{\de e^1}$,
$K_2\superx=e^1\frac{\de\superx}{\de e^2}$ and $K_2\supereta=e^1\frac{\de\supereta}{\de e^2}$,
we get
\begin{align*}
\delta_1 &= \eta_1^{+\mu}\frac\de{\de x^\mu} + b^{+\mu}\frac\de{\de\eta_2^{+\mu}} + \eta_{1\mu}\frac\de{\de b_\mu} - x^+_\mu\frac\de{\de \eta_{2\mu}},\\
\delta_2 &= \eta_2^{+\mu}\frac\de{\de x^\mu} - b^{+\mu}\frac\de{\de\eta_1^{+\mu}} + \eta_{2\mu}\frac\de{\de b_\mu} + x^+_\mu\frac\de{\de \eta_{1\mu}},\\
K_1 &= \eta_{1\mu}\frac\de{\de\eta_{2\mu}}+\eta_1^{+\mu}\frac\de{\de \eta_2^{+\mu}},\\
K_2 &=\eta_{2\mu}\frac\de{\de\eta_{1\mu}}+\eta_2^{+\mu}\frac\de{\de \eta_1^{+\mu}}.
\end{align*}
The corresponding Hamiltonian functions, with respect to the symplectic structure
\[
\omega=\int\dd e^2\dd e^1\; \delta\superx^\mu\delta\supereta_\mu=
\delta x^\mu\delta x^+_\mu +\delta\eta_1^{+\mu}\delta\eta_{2\mu} - \delta\eta_2^{+\mu}\delta\eta_{1\mu}-\delta b^{+\mu}\delta b_\mu,
\]
are given by
\begin{align*}
\tau_1&= x^+_\mu\eta_1^{+\mu}-\eta_{1\mu}b^{+\mu},\\
\tau_2&= x^+_\mu\eta_2^{+\mu}-\eta_{2\mu}b^{+\mu},\\
\rho_1 &= -\eta_{1\mu}\eta_1^{+\mu},\\
\rho_2 &= \eta_{2\mu}\eta_2^{+\mu}.
\end{align*}
\end{Rem}

We now turn back to the effective action. It turns out that it is not only $\delta_1$- and $\delta_2$\ndash closed, but actually exact:
\[
\Seff=\delta_2\delta_1\sigma,\quad \sigma:=\frac12\pi^{\mu\nu}(x)b_\mu b_\nu.
\]
{}From all the above it follows that $\Seff$ Poisson commutes not only with $\tau_1$ and $\tau_2$, but also with $\rho_1$ and $\rho_2$. Notice that the Jacobi identity
for $\pi$ implies $(\Seff,\sigma)=0$.

Now consider the regularized effective action
\[
\Seff^{\epsilon,t_1,t_2}:= \epsilon \Seff -\ii\hbar( t_1\tau_1 +t_2\tau_2),
\]
which satisfies the QME for all $\epsilon,t_1,t_2$. By all the above it follows that
\begin{align*}
\frac\de{\de t_1}\EE^{\frac\ii\hbar\Seff^{\epsilon,t_1,t_2}} &=\Delta\left( \frac1{t_2}\EE^{\frac\ii\hbar\Seff^{\epsilon,t_1,t_2}} \rho_1\right),\\
\frac\de{\de t_2}\EE^{\frac\ii\hbar\Seff^{\epsilon,t_1,t_2}} &=\Delta\left( \frac1{t_1}\EE^{\frac\ii\hbar\Seff^{\epsilon,t_1,t_2}} \rho_2\right),
\end{align*}
which means that, as long as the parameters $t_1$ and $t_2$ are different from zero, the regularized effective action is independent of them up to quantum canonical transformations. We also have
\[
\frac\de{\de\epsilon}\EE^{\frac\ii\hbar\Seff^{\epsilon,t_1,t_2}} = \Delta\left( \frac\ii{\hbar t_2}\EE^{\frac\ii\hbar\Seff^{\epsilon,t_1,t_2}} \delta_1\sigma\right)=
-\Delta\left( \frac\ii{\hbar t_1}\EE^{\frac\ii\hbar\Seff^{\epsilon,t_1,t_2}} \delta_2\sigma\right),
\]
which implies that, as long as one of the parameters $t_1$ and $t_2$ is different from zero, the regularized effective action is independent of $\epsilon$ up to quantum canonical transformations. This in particular means that  the partition function is independent of $\epsilon$ and that, in order to compute it, we may simply set $\epsilon$ to zero.

To perform the final computation we further deform the regularized effective action by adding one more irrelevant term.
Namely, let $G$ be a function on $\Map(\bbR^2[1],M)$. Then
\[
\Seff^{0,t_1,t_2,G}:= -\ii\hbar( t_1\tau_1 +t_2\tau_2+\delta_2\delta_1 G)
\]
satisfies the QME for all $t_1,t_2,G$. Moreover, if we take a path $G(t)$ of such functions, we get
\[
\frac\de{\de t}\EE^{\frac\ii\hbar\Seff^{0,t_1,t_2,G(t)}} = \Delta\left( \frac{\EE^{\frac\ii\hbar\Seff^{0,t_1,t_2,G(t)}} \delta_1\dot G(t)}{t_2}\right)=
-\Delta\left( \frac{\EE^{\frac\ii\hbar\Seff^{0,t_1,t_2,G(t)}} \delta_2\dot G(t)}{t_1}\right),
\]
which means that, as long as one of the two parameters $t_1$ and $t_2$ is different from zero,
adding the new term is irrelevant up to quantum canonical transformations.
We are now ready to compute the partition function. Namely, we choose
$\calL:=\Map(\bbR^2[1],M)$
as the Lagrangian submanifold of $\Map(\bbR^2[1],T^*[1]M)$ over which we integrate. Since $\tau_1|_\calL=\tau_2|_\calL=0$, we get
\[
Z = \int_\calL \EE^{\frac\ii\hbar\Seff^{0,t_1,t_2,G}}  = \int_{\Map(\bbR^2[1],M)} \EE^{\delta_2\delta_1 G}
\]
and we already know that the last integral is independent of $G$. We only have to make sure that $G$ is chosen is such a way that the integral is well defined (choosing $G=0$, e.g.,
would lead to $\infty\cdot 0$). A good choice is $G:=g_{\mu\nu}(x)\eta_1^{+\mu}\eta_2^{+\nu}$ where $g^{\mu\nu}$ is a Riemannian metric on target.
An explicit computation \cite{BE,JF} then shows that $Z=\chi(M)$.

\begin{Rem}
Switching $\epsilon$ to zero first and then turning on the regularizing term in $G$ is a bit formal since we pass through the solution to the QME where both terms are absent. This solution has a singular integral (of the type $0\cdot\infty$) on $\calL$. In order to find a non formal regularization it is necessary to have additional structure on the Poisson manifold. For instance let us look for $G$ such that $\Seff^{\epsilon,t_1,t_2,G}$ satisfies the QME for any $\epsilon$, preserving the property that the variation of $G$  produces a quantum canonical transformation. Indeed, let us assume $G$ as above but let $g$ be possibly degenerate. If $(\Seff,G)=0$ then $\Seff^{\epsilon,t_1,t_2,G}$ satisfies the
QME and the change of $G$ is a quantum canonical transformation. This property is equivalent to require that $\pi\circ g=0$ and $L_{V}g = 0$ for every vector field $V$ tangent to the symplectic leaves. In special cases we may find such a $g$ and in addition a K\"ahler structure on the leaves, compatible with the symplectic structure,
such that a gauge fixing given by a mixture of what we discussed in this subsection and the K\"ahler one is available.
We plan to investigate  the geometrical conditions needed for this gauge fixing in the future.

\end{Rem}

\subsection{Regularization on the space of fields}\label{ss-rsf}
The argument of the previous subsection may formally be lifted to the space of fields to show that the regularized action is actually independent, up to quantum canonical transformations, of the Poisson structure.
Let $s^1$ and $s^2$ denote the coordinates on the two $S^1$ factors of the torus $\T^2$,
and let $e^1$ and $e^2$ denote the corresponding fiber coordinates on $T[1]\T^2$. We now denote by $\delta_1$, $\delta_2$,
$K_1$ and $K_2$ the lifts of the vector fields
$\frac\de{\de e^1}$,  $\frac\de{\de e^2}$, $e^2\frac\de{\de e^1}$ and $e^1\frac\de{\de e^2}$
to  the space of fields
$\calF=\Map(T[1]\T^2,T^*[1]M)$.
We denote by $\tau_1$, $\tau_2$, $\rho_1$ and $\rho_2$ their Hamiltonian functions. They satisfy \eqref{eone} and \eqref{etwo}, and formally also \eqref{ethree}.

\begin{Rem}
For completeness, we give explicit expressions also in this case, even if we are not going to need them.
If we write
\begin{align*}
\sfX &= X + \eta_1^+ e^1 + \eta_2^+e^2 + \beta^+ e^1 e^2,\\
\sfeta &= \beta + \eta_1 e^1 + \eta_2 e^2 + X^+ e^1 e^2,
\end{align*}
we then have
\begin{align*}
\delta_1 X&=-\eta_1^+, &\delta_1\eta_2^+&=\beta^+,&\delta_1\beta&=\eta_1,&\delta_1\eta_2&=-X^+,\\
\delta_2 X&=-\eta_2^+,&\delta_2\eta_1^+&=-\beta^+,&\delta_2\beta&=\eta_2,&\delta_2\eta_1&=X^+,
\end{align*}
and
\begin{align*}
K_1\eta_2^+ &= \eta_1^+, & K_2\eta_1^+&=\eta_2^+,\\
K_1\eta_2&=\eta_1, & K_2\eta_1&=\eta_2.
\end{align*}
With respect to the symplectic form
\[
\Omega = \int_{\calF}
\delta\sfX\delta\sfeta =
\int_{\T^2} (\delta X^\mu\delta X^+_\mu - \delta\eta_1^{+\mu}\delta\eta_{2\mu} +\delta\eta_2^{+\mu}\delta\eta_{1\mu}+\delta\beta^{+\mu}\delta\beta_\mu)
\,\dd s^1\dd s^2,
\]
the corresponding Hamiltonian functions are
\begin{align*}
\tau_1&= \int_{\T^2} (-\eta_1^{+\mu}X^+_\mu+\beta^{+\mu}\eta_{1\mu})\,\dd s^1\dd s^2,\\
\tau_2&= \int_{\T^2} (-\eta_2^{+\mu}X^+_\mu+\beta^{+\mu}\eta_{2\mu})\,\dd s^1\dd s^2,\\
\rho_1 &= \int_{\T^2}\eta_1^{+\mu}\eta_{1\mu}\,\dd s^1\dd s^2,\\
\rho_2 &= -\int_{\T^2}\eta_2^{+\mu}\eta_{2\mu}\,\dd s^1\dd s^2.
\end{align*}
Notice that despite their non covariant look the above formulae are actually globally well defined.
\end{Rem}

The action $S=S_0+S_\pi$ is  $\delta_1$- and $\delta_2$\ndash closed; it turns out that the interaction part
$S_\pi$ is actually exact:
\[
S_\pi=\delta_2\delta_1\sigma_\pi,\quad \sigma_\pi:=\int_{\T^2}\frac12\pi^{\mu\nu}(X)\beta_\mu \beta_\nu\,\dd s^1\dd s^2.
\]
{}From all the above it follows that $S$ Poisson commutes not only with $\tau_1$ and $\tau_2$, but also with $\rho_1$ and $\rho_2$. Notice that the Jacobi identity
for $\pi$ implies $(S,\sigma_\pi)=0$.

Now consider the regularized  action
\[
S^{\epsilon,t_1,t_2}:= S_0+\epsilon S_\pi -\ii\hbar( t_1\tau_1 +t_2\tau_2),
\]
which satisfies the CME and formally also the QME for all $\epsilon,t_1,t_2$. By all the above it follows that, formally,
\begin{align*}
\frac\de{\de t_1}\EE^{\frac\ii\hbar S^{\epsilon,t_1,t_2}} &=\Delta\left( \frac1{t_2}\EE^{\frac\ii\hbar S^{\epsilon,t_1,t_2}} \rho_1\right),\\
\frac\de{\de t_2}\EE^{\frac\ii\hbar S^{\epsilon,t_1,t_2}} &=\Delta\left( \frac1{t_1}\EE^{\frac\ii\hbar S^{\epsilon,t_1,t_2}} \rho_2\right),
\end{align*}
which means that, as long as the parameters $t_1$ and $t_2$ are different from zero, the regularized action is independent of them up to quantum canonical transformations. We also have, again formally,
\[
\frac\de{\de\epsilon}\EE^{\frac\ii\hbar S^{\epsilon,t_1,t_2}} = \Delta\left( \frac\ii{\hbar t_2}\EE^{\frac\ii\hbar S^{\epsilon,t_1,t_2}} \delta_1\sigma\right)=
-\Delta\left( \frac\ii{\hbar t_1}\EE^{\frac\ii\hbar S^{\epsilon,t_1,t_2}} \delta_2\sigma\right),
\]
which implies that, as long as one of the parameters $t_1$ and $t_2$ is different from zero, the regularized action is independent of $\epsilon$ up to quantum canonical transformations.  This in particular means that  the partition function is independent of $\epsilon$ and that, in order to compute it,
we may simply set $\epsilon$ to zero. It is now easy to see that, for a reasonable choice of propagators, the effective action for $S^{0,t_1,t_2}$
is simply the restriction to vacua, that is the the regularized effective action $\Seff^{0,t_1,t_2}$ considered in the previous subsection.

\section{Globalization of the effective action}\label{s:global}
We now go back to the problem of globalizing $\tSeff^{(0)}$ in the general case.
Recall that $\tSeff\in\Gamma(\Lambda T^*M\otimes\Hat S\calH^*[[\hbar]])$ satisfies the modified QME \eqref{e:mQME}. We write
$\tSeff=\sum_{i=0}^m \tSeff^{(i)}$, where $\tSeff^{(i)}$ is the $i$\ndash form component and $m=\dim M$. In form degree zero, we have
\[
\frac12(\Tilde S_\text{eff}^{(0)},\Tilde S_\text{eff}^{(0)})-\ii\hbar\Delta\Tilde S_\text{eff}^{(0)}=0,
\]
which is the usual QME.

The modified QME is preserved under modified quantum canonical transformations. Namely, $T\in\Gamma(\Lambda T^*M\otimes\Hat S\calH^*[[\hbar]])$ of total degree $-1$ generates the infinitesimal transformation
\[
\delta\tSeff = \dd_x T +(\tSeff,T)-\ii\hbar\Delta T
\]
which preserves the modified QME at the infinitesimal level. Notice that, setting $T=\sum_{i=0}^m T^{(i)}$, we get in form degree zero
\[
\delta\tSeff^{(0)} = (\tSeff^{(0)},T^{(0)})-\ii\hbar\Delta T^{(0)},
\]
which is a usual infinitesimal quantum canonical transformation. The goal of this Section is to prove the following
\begin{Thm}\label{t:glob}
There is a quantum canonical transformation starting at order $1$ in $\hbar$
that makes the form degree zero part $\tSeff^{(0)}$ of the effective action closed with respect to the induced Grothendieck differential
$D=\dd_x + (S_R|_\text{vacua},\ )$
on $\Gamma(\Lambda T^*M\otimes\Hat S\calH^*[[\hbar]])$,
where $S_R|_\text{vacua}$ denotes the evaluation of $S_R$ on vacua.
\end{Thm}

This will ensure that the so obtained effective action, call it $\Check S_\text{eff}^{(0)}$, is the image under $T\phi^*$ of a global effective action $\Seff$. Since
$\Check S_\text{eff}^{(0)}\in\Gamma(\Hat S\calH^*[[\hbar]])$, it follows from the discussion just before subsection~\ref{ss:gauge} that
$\Seff$ is a section of
$\Hat S\Tilde\calH^*[[\hbar]]$,\footnote{Recall that  $\Tilde\calH_x=H^{>0}(\Sigma)\otimes T_xM\oplus H^\bullet(\Sigma)\otimes T^*_xM[1]$.}
 i.e., a formal power series in $\hbar$ of functions on $\Tilde\calH$ (formal in the fiber coordinates).
Again, we may identify $\Tilde\calH$ with the canonical global moduli space of vacua by using a connection (e.g., the one in $\phi$). By Remark~\ref{r:changephi}, we conclude that the class of $\Seff$
under quantum canonical transformations is a well-defined object independent of all choices.

\subsection{Proof of Theorem~\ref{t:glob}}
We start with a simple observation:
\begin{Lem}\label{l:tree}
Write $\tSeff^{(i)}=\sum_{k=0}^\infty \hbar^k S^{(i)}_k$.
If the propagator satisfies the properties in \eqref{K properties}, then
$S^{(i)}_0=0$ $\forall i>1$, whereas $S^{(0)}_0$  and $S^{(1)}_0$ are
obtained by the evaluation on vacua of $\Hat S$ and $S_R$, respectively.
\end{Lem}
\begin{proof}
The terms for $k=0$ correspond to trees in the expansion in Feynman diagrams; so,
using the notations of Section~\ref{s:caxial}, what we have to prove is that
the source part $W_\Gamma^\mr{source}$ vanishes for any tree $\Gamma$ containing more than one vertex.

This is checked by the following degree counting argument. Consider a tree $\Gamma$ containing more than one vertex. Let $V_k$ be the number of vertices in $\Gamma$ of internal valence (i.e., not counting the leaves) equal to $k\geq 1$. Then the total number of vertices is
$$V=\sum_{k\geq 1}V_k,$$
the number of internal edges is
\be E=\frac{1}{2}\sum_{k\geq 1}k V_k \label{E via V_k},\ee
and the Euler characteristic of $\Gamma$ is
\be 1=V-E=\sum_{k\geq 1}\frac{2-k}{2}V_k \label{1=V-E}.\ee

Next, the source part $W_\Gamma^\mr{source}$ vanishes automatically due to $K\PP=\PP K=0$ and $K^2=0$, unless the following two properties hold for the decoration of leaves by cohomology classes $\chi_\alpha\in H^\bt(\Sigma)$:
\begin{enumerate}[(i)]
\item At every vertex of internal valence 1 there are at least two incident leaves decorated by cohomology classes of non-zero degree. (Otherwise $W_\Gamma^\mr{source}$ vanishes due to $K\PP=0$.) \label{dc condition at 1-val vertices}
\item At every vertex of internal valence 2 there is at least one incident leaf decorated by a cohomology class of non-zero degree. (Otherwise $W_\Gamma^\mr{source}$ vanishes due to $K^2=0$.) \label{dc condition at 2-val vertices}
\end{enumerate}
This gives a lower bound $E+2 V_1+V_2$ for the form degree of the integrand in (4.5); since it should coincide with the dimension of the space $\Sigma^{\times V}$ it is integrated against, we have the inequality
\be E + 2 V_1 + V_2 \leq 2 V.\ee
By (\ref{E via V_k}) this is equivalent to
$$\frac{1}{2}\,V_1+\sum_{k\geq 3}\frac{k-4}{2}\,V_k \leq 0.$$
Subtracting (\ref{1=V-E}) from this inequality, we get
$$\sum_{k\geq 3}(k-3)V_k\leq -1$$
which is a contradiction. Thus it is impossible to find a decoration of leaves of $\Gamma$ satisfying properties (\ref{dc condition at 1-val vertices}) and (\ref{dc condition at 2-val vertices}) simultaneously. Therefore, $W_\Gamma^\mr{source}$ vanishes for any decoration.
\end{proof}
We now set $S^{(1)'}=\tSeff^{(1)}-S^{(1)}_0$ and $S^{(i)'}=\tSeff^{(i)}$ for $i>1$.
\begin{Lem}
There is a modified quantum canonical transformation starting at order $1$ in $\hbar$
after which all $S^{(i)'}$ for $i\ge1$ vanish.
\end{Lem}
\begin{proof}
We work by induction on the order of $\hbar$. At order zero the statement holds by Lemma~\ref{l:tree}.
Assume that $S^{(i)'}_r=0$ $\forall i\ge 1$ and $\forall r<k$.
Then the modified quantum master equation yields the identities
\[
DS^{(i)'}_k+\Omega S^{(i+1)'}_k=0,\ \forall i\ge1,
\]
with $D:=\dd_x+(S^{(1)}_0,\ )$ and $\Omega:=(S^{(0)}_0, \ )$. We already know that $D^2=\Omega^2=0$ and that $D$ and $\Omega$ commute.
If $T=\sum_{r=k}^\infty \hbar^rT_r$ is a generator starting at the order $k$, we then have the infinitesimal transformations
\[
\delta S^{(i)'}_k = D T^{(i-1)}_k + \Omega T^{(i)}_k,\ \forall i\ge 1.
\]
By dimensional reasons $DS^{(m)'}=0$, with $m=\dim M$, and since the $D$\ndash cohomology is concentrated in degree zero, we can find a
$\tau\in \Gamma(\Lambda T^*M\otimes\Hat S\calH^*)$ such that $S^{(m)'}=D\tau$. We now consider the transformation with generator $T=-\hbar^k\tau$.
Hence we get $\delta S^{(m)'}_k=-D\tau$, $\delta S^{(m-1)'}_k=-\Omega\tau$ and $\delta S^{(i)'}_k=0$ for $i<m-1$.
Integrating this transformation up to time $1$, we make $S^{(m)'}_k$ vanish; as a result the new $S^{(m-1)'}_k$ will be $D$\ndash closed. We may then proceed like this until we make all
the $S^{(i)'}$ vanish. This proves our claim.

Notice that these transformations may change the $S^{(i)'}_r$ for $r>k$. Moreover, the
generator used to kill $S^{(1)'}_k$ will act on $S^{(0)}_r$ for $r\ge k$ by a quantum canonical transformation.
\end{proof}
This completes the proof of Theorem~\ref{t:glob}.

As a final remark, observe that in the case when $\pi$ is regular and unimodular we start with $S^{(1)'}=0$, so we have two different but equivalent ways of getting the global action. One consists in taking the original
$\tSeff^{(0)}$, the other in applying the method described in this Section since nothing guarantees that the remaining $S^{(i)'}$ vanish at the start. After applying the method we get another effective action $\Check S_\text{eff}^{(0)}$
that simply differs from $\tSeff^{(0)}$  by a quantum canonical transformation and is also the image of $T\phi^*$ of a global action which we denote by
$\Check S$. Eventually, the two global effective actions $S$ and $\Check S$ simply differ by a quantum canonical transformation starting at order $\hbar$.

\section{Conclusions and perspectives}
In this paper
we have studied the effective action of the Poisson sigma model
on a closed surface $\Sigma$, where the Poisson structure $\pi$ on the target $M$ is treated perturbatively and,
for consistency, has to be assumed to be unimodular unless $\Sigma$ is a torus. We have shown how to obtain a
global effective action $\Seff$ as an $\hbar$\ndash dependent function on the moduli space of vacua
of the theory with zero Poisson structure, around which we are perturbing.
Because of the freedom in the choice of gauge fixing and the details of
globalization, $\Seff$ is, as usual, only well-defined up to quantum canonical transformations. By a reasonable
choice of the class of allowed gauge fixings---namely, those for which the propagator enjoys
properties \eqref{K properties}---we make sure that the order zero $S_\text{eff,0}$ of the effective action
is fixed and equal to the evaluation on vacua of the Poisson\ndash dependent part $S_\pi$ of the action; moreover,
the remaining quantum canonical transformations will start at order $1$.

In the cases when $\Sigma$ is a torus or $\pi$ is regular and unimodular, we have shown that 
$\Seff$ has (a representative with)
no quantum corrections. In the particular case when $\Sigma$ is a torus, $\pi$ is nondegenerate and
there is a compatible complex structure, we can use the latter to gauge-fix the remaining integration over vacua:
the final result is that, as expected from the Hamiltonian formulation and from comparison with with the A-model,
the partition function is the Euler characteristic of the target. An alternative approach that produces the same result consists in regularizing the effective
action by adding the Hamiltonian functions of supersymmetry generators.
In general, the effective action modulo quantum canonical transformations is an invariant of the Poisson
structure.

Recall that each order in $\hbar$ of $\Seff$ is actually a section of a vector bundle
$\calZ_g:=\Hat S\Tilde\calH^*$ over the target $M$ whose structure is fixed by the genus $g$ of the source
$\Sigma$.
These sections are just tensors of a particular sort.
The lowest order in the quantum master equation for $\Seff$ implies that  $S_\text{eff,0}$
solves the classical master equation, i.e., that it defines a differential on
$\Gamma(\calZ_g)$, which we call the genus $g$ Poisson complex.
Since $S_\text{eff,0}$ is also $\Delta$\ndash closed, the lowest nonvanishing quantum contribution to
$\Seff$ is a cocycle in the genus $g$ Poisson complex\footnote{If this happens at order $1$,
which can be easily seen not to be the case for $g\le1$, this class is also what is left after modding out by
quantum canonical transformations at this order.}
and defines an invariant of the Poisson structure,
which might be possible to compute in concrete examples.


\end{document}

